\documentclass[twocolumn,amsmath,amssymb,aps,prx]{revtex4-2}
\usepackage{graphicx}
\usepackage{dcolumn}
\usepackage{bm}
\usepackage{quantikz}
\usepackage{float}
\makeatletter
\let\newfloat\newfloat@ltx
\makeatother
\usepackage{algorithm}
\usepackage{algcompatible}
\usepackage{amsmath,amssymb}
\usepackage{amsthm}
\usepackage{comment}
\usepackage[colorlinks=true, allcolors=blue]{hyperref}
\floatname{algorithm}{Protocol }
\newtheorem{theorem}{Theorem}[section]
\newtheorem{lemma}[theorem]{Lemma}

\usepackage{xcolor}
\usepackage{soul}

\begin{document}
\preprint{APS/123-QED}
\title{Verifiable End-to-End Delegated Variational Quantum Algorithms}
\author{Matteo Inajetovic$^{1}$, Petros Wallden$^{2}$ and Anna Pappa$^{1}$}
\affiliation{$^1$ Electrical Engineering and Computer Science Department, Technische Universit{\"a}t Berlin}
\affiliation{$^2$ Quantum Software Lab, School of Informatics, University of Edinburgh}
\date{\today}

\begin{abstract}
\noindent Variational quantum algorithms (VQAs) have emerged as promising candidates for solving complex optimization and machine learning tasks on near-term quantum hardware. However, executing quantum operations remains challenging for small-scale users because of several hardware constraints, making it desirable to delegate parts of the computation to more powerful quantum devices. In this work, we introduce a framework for delegated variational quantum algorithms (DVQAs), where a client with limited quantum capabilities delegates the execution of a VQA to a more powerful quantum server. In particular, we introduce a protocol that enables a client to delegate a variational quantum algorithm to a server while ensuring that the input, the output and also the computation itself remain secret. Additionally, if the protocol does not abort, the client can be certain that the computation outcome is indeed correct.
This work builds on the general verification protocol introduced by Fitzimons and Kashefi (2017), tailoring it to VQAs.
Our approach first proposes a verifiable protocol for delegating the quantum computation at each optimization step of a VQA, and then combines the iterative steps into an error-resilient optimization process that offers end-to-end verifiable algorithm execution.
We also simulate the performance of our protocol tackling the Transverse Field Ising Model. Our results demonstrate that secure delegation of variational quantum algorithms is a realistic solution for near-term quantum networks, paving the way for practical quantum cloud computing applications.
\end{abstract}

\maketitle

\section{Introduction}
Quantum computing is poised to revolutionize information processing by leveraging the principles of quantum mechanics to solve problems that are intractable for classical computers. With advancements in quantum hardware, software, and algorithms, we are witnessing rapid progress toward practical quantum computation. However, the realization of fault-tolerant, large-scale quantum computers remains a long-term goal. In the near term, noisy intermediate-scale quantum (NISQ) devices—quantum processors with tens to hundreds of qubits but without robust error correction—represent the frontier of quantum technology. These NISQ devices offer promising opportunities for obtaining a quantum advantage in specific computational tasks, such as quantum simulation, optimization, and machine learning. Yet, their inherent noise and limited coherence times present significant challenges.

Variational quantum algorithms (VQA) have emerged as a promising approach to leverage NISQ devices to solve complex computational problems \cite{Cerezo_2021}. They are hybrid quantum-classical algorithms, in which a quantum processor executes a circuit with adjustable parameters, and a classical optimizer iteratively updates these parameters to minimize a given cost function. By combining quantum circuits with classical optimization techniques, VQAs can efficiently explore high-dimensional solution spaces while mitigating the limitations of NISQ hardware.

As quantum processors become more accessible via cloud-based platforms, it is natural to consider the concept of delegated variational quantum algorithms (DVQAs). Such a process is envisioned to enable a classical client to delegate complex quantum computational tasks to a remote quantum server while maintaining control over the optimization process. Ensuring the privacy and correctness of delegated computations is crucial, particularly in scenarios where the client lacks direct access to the quantum hardware.

Most delegated quantum computation protocols are instantiated in the measurement-based quantum computing (MBQC) model, which naturally allows for blindness and verifiability of the computation \cite{broadbent_universal_2009, Barz_2013, UVBQC17,  gheorghiu_verification_2019}. Compared to the circuit model, MBQC inherently supports these features due to the one-way structure and the decoupling of resource state preparation from the actual (measurement-based) computation. Moreover, MBQC is particularly well suited to photonic implementations and quantum network architectures, further reinforcing its advantages for delegated quantum computation.

In the context of variational quantum algorithms, a direct application of the aforementioned established delegation methods is not feasible. Due to the hybrid nature of VQAs, which involves classical optimization  and non-deterministic quantum computations, a tailored framework for delegating such algorithms is required. This article explores the potential of measurement-based delegated variational quantum algorithms (MB-DVQAs). While
previous studies have already provided the first investigations of MB-DVQAs \cite{li_quantum_2021, shingu_variational_2022, wang10039146, YANG2025131474}, they focused on the shot-level verifiability of single iterations of the quantum variational procedure and therefore did not achieve verifiability for the entire optimization process.
More specifically, given the number of shots required by VQAs, verifying every shot is not efficient.

We give an alternative protocol to achieve individual optimization step verifiability, and then propose a verifiable end-to-end MB-DVQA,
encompassing both quantum and classical computations involved.

\subsection*{Overview of the contributions}
\begin{enumerate}
    \item We show how to verify one step of an MB-DVQA (Protocol \ref{prot:VDgrad}) by using a trapification scheme that alternates computation and test rounds \cite{leichtle_verifying_2021}, and prove that it succeeds with probability that grows exponentially with the number of total rounds (Theorem \ref{th:ver_step}).
    \item We propose a new MB-DVQA using a custom gradient-descent optimizer (Protocol \ref{prot:VDvqa}) that reruns highly attacked iterations without aborting the full optimization and show that it is verifiable (Theorem \ref{th:ver_vqa}). 
    \item We benchmark the custom optimizer outlined in Protocol \ref{prot:VDvqa} through numerical simulations of a Variational Quantum Eigensolver (VQE) \cite{Peruzzo_2014} that finds the ground state energy of the  Transverse Field Ising Hamiltonian, incorporating simulated gradient perturbations that mimic adversarial attacks.
\end{enumerate}

\section{Background}
A central component in VQAs is a Parametrized Quantum Circuit (PQC). This circuit (also called Ansatz), is associated with a unitary $U(\theta)$, where $\theta$ is a vector of parametrized angles of dimension $N_P$. The main goal of a VQA is to minimize a cost function $f$ represented as:
\begin{align*}
        f(\theta)=Tr[OU(\theta)\rho_{in}U^{\dagger}(\theta)]
\end{align*}
where $\rho_{in}$ is the initial state and $O$ is the observable describing the problem to solve. In particular, $O$ is decomposed as sum of $N_o$ local Pauli operators $P_i$:

\[
O = \sum_{i=1}^{N_o} O_i=\sum_{i=1}^{N_o}c_i P_i
\]

We will focus on gradient-based optimization, making use of the Parameter Shift Rule \cite{Mitarai_2018,PhysRevA.99.032331} to compute an approximation $g$ of the gradient $\nabla f$. For the resulting vector $g=(g_1(\theta),\dots,g_{N_P}(\theta))$, it holds that for every element of the gradient vector the partial derivative is 
$\frac{\partial f}{\partial \theta}
=\frac{1}{2}[f_{+}-f_{-}]
$, where:
\[
f_{+}=f(\theta+\frac{\pi}{2}) \text{ and } f_{-}=f(\theta-\frac{\pi}{2}).
\]
This means that for a given PQC with $N_P$ parameters, the total number of function evaluations per step is $2N_P$. Using a fixed amount of measurement shots per function evaluation $N_s$, we require a total number of $2N_PN_s$ shots to perform one \emph{step} of the optimization process.

Till recently, all proposals of VQEs were based on the circuit model of computation. In the last years, variational quantum algorithms based on measurement-based quantum computation have also been proposed \cite{ferguson_measurement-based_2021,proietti_native_2022,schroeder_deterministic_2023,Majumder_2024,stollenwerk2024measurementbasedquantumapproximateoptimization, Marqversen_2023,Qin_2024}. MBQC \cite{raussendorf_one-way_2001} differs significantly from the conventional gate-based approach; instead of executing quantum gates sequentially, MBQC begins with the preparation of a highly entangled state, and the computation is then carried out through a sequence of adaptive quantum measurements that depend on the outcomes of the previous measurements. 

It has been shown that a client can securely and blindly delegate a quantum computation to a server by appropriately encrypting their input and the measurement angles of the MBQC pattern \cite{broadbent_universal_2009,fitzsimons2016privatequantumcomputationintroduction}.
However, another crucial property in delegated quantum computing is verifiability: the Client must be able to check whether the delegated computation has been executed correctly without malicious actions affecting the desired outcome. More specifically, the delegated protocol either outputs the desired computation or aborts. To ensure this property, it is sufficient to find an upper bound on the probability of failure, that is, the probability of having an incorrect quantum computation without detecting it by returning an abort signal.
Initially proposed in \cite{UVBQC17}, trapification allows to detect dishonest behavior by modifying the graph $G$ associated with the delegated computation, by inserting randomly prepared trap qubits in $\ket{+_{\theta}}=(\ket{0}+e^{i \theta}\ket{1})/\sqrt{2}$ and separating those from the rest of the computation by means of so-called dummy qubits, randomly chosen from $\{\ket{0},\ket{1}\}$.

In a previous study of delegated bounded-error quantum polynomial-time
computation \cite{leichtle_verifying_2021}, the proposed verification scheme interleaves $d$ \emph{computation rounds} and $t$ \emph{test rounds} (for a total of $n=d+t$ rounds), rather than inserting traps into the graph $G$ associated with the desired computation. The latter approach results in an enlarged graph, whereas interleaving rounds maintain the original graph size.  In the computation rounds, the Client prepares and sends the qubits and the measurements associated to the desired computation unitary $U$ (defined by graph $G$), while the test rounds instead are just made up of traps and dummy qubits. In particular, \cite{leichtle_verifying_2021} makes use of a $c$-coloring of $G$ to place traps and dummy qubits. By having the same graph structure for both computation and test rounds, the Server cannot distinguish them due to blindness inherited by previous works \cite{broadbent_universal_2009}.

\section{Methods}
Our goal is to provide an MB-DVQA that converges and is at the same time blind, secure and verifiable. Since blindness and security are directly inherited from previous works on delegated MBQC \cite{Dunjko_2014}, the only thing that remains to address is verifiability. We first propose a verifiable protocol (Protocol \ref{prot:VDgrad}) computing the gradients of one MB-DVQA step using the parameter shift rule. In particular, we provide a bound (Lemma \ref{lem:rel_error}) on the gradient's relative error in case of adversarial behavior, namely when cost function evaluations are affected by a certain number of corrupted computation rounds. Subsequently, we make use of this result to come up with a trapification scheme based on \cite{leichtle_verifying_2021}, yielding an undetected highly corrupted outcome with a very low probability (Theorem \ref{th:ver_step}).
We then exploit this protocol within the context of a complete MB-DVQA using gradient descent to minimize the cost function. Specifically, we modify the classical optimization part of the VQA by re-performing highly attacked steps detected by the step-level protocol and prove convergence and verifiability (Theorem \ref{th:ver_vqa}) for the full MB-DVQA protocol (Protocol \ref{prot:VDvqa}).

\subsection{Verifiability for one step of the MB-DVQA}
\label{sec:verif_pqc}

Existing blind and verifiable MBQC protocols can in principle be directly applied to VQAs at shot-level computation. Nevertheless, this approach is impractical given the large number of shots required by the total number of function evaluations in the variational procedure. Furthermore, running a verification protocol for each shot and abort every corrupted one implies a heavy consumption of runtime resources. 
The robust VBQC protocol proposed by \cite{leichtle_verifying_2021} allows to delegate multiple shots at once and appears as a better candidate to delegate variational algorithms. However this protocol, designed for deterministic BQP computations, does not account for the probabilistic, expectation-value-based cost functions required for variational parameter updates. Consequently, adapting this protocol to VQAs requires a tailored verification framework that aggregates measurements across shots and handles the multiple evaluations required per optimization step, while tolerating a small, non-critical fraction of corrupted shots.
Therefore the aim of Protocol  \ref{prot:VDgrad} is to delegate the quantum computation required to compute the gradient of the cost function with a relative error bounded by a target value  $\varepsilon_{th}$. This quantum computation consists of all the $2N_P$ evaluations of the cost function $f$, each performed $N_s$ times.

To provide verifiability for this protocol and allow the client to check for the correctness of the computation, we adapt the trapification protocol proposed in \cite{leichtle_verifying_2021}. In that work, majority vote was used to bound the probability of failure of the verification process for BQP computations. Here instead, due to the probabilistic nature of PQCs, we consider the computation outcome through the expectation value of observables, which are typically used to evaluate the cost function. We therefore execute $d=2N_PN_s$ computation rounds and $t$ test rounds shared across different quantum computations (i.e. distinct $f$ evaluations). This allows us to reduce the amount of required test rounds compared to the case where each quantum computation is delegated separately.

\begin{algorithm}
\caption{\raggedright Verifiable Step of the MB-DVQA} 
\begin{algorithmic}[1]
  \STATE \textbf{Inputs}: $G$ and associated $c$-coloring, $\theta=\{\theta_i\}_{i=1}^{N_P}$, $N_s$, $\varepsilon_{th}$.
    \STATE \textbf{Protocol:}
    \begin{enumerate}
       \item Suppose $n=d+t$ the number of rounds, where $d=2 N_PN_s$ are the computation and $t$ are the test rounds and  the Client chooses at random which round belongs in each set. For each round:
        \begin{itemize}
            \item if it is a test round, the client chooses a random graph coloring and sends the trap and dummy qubits to the server.
            \item if it is a computation round, the client sends the qubits associated to $G$.
            \item the server performs the blind measurement based computation and returns the output.
        \end{itemize}
        \item The client gathers and decodes the server's output:
        \begin{itemize}
             \item in case she detects a trap failure, she \textbf{aborts}.
             \item otherwise, she \textbf{accepts}, obtaining $g(\theta)$ with a relative error up to $\varepsilon_{th}$.
        \end{itemize}  
    \end{enumerate}    
\end{algorithmic}
\label{prot:VDgrad}
\end{algorithm}


In order to prove verifiability for our protocol, we first establish a connection between the relative error in the gradient approximation and the adversarial attacks, represented in our framework by means of the number of corrupted computation rounds $\delta$.

\begin{lemma}
    We assume a DVQA step performing a Parameter Shift Rule to compute a gradient approximation $g$ of $\nabla f$ with a global observable $O=\sum_{i}c_iP_i$. Using $N_s$ shots per circuit evaluation,
    the gradient relative error $\varepsilon = \frac{\|g-\hat{g}\|_2}{\|g\|_2}$ and the number of corrupted computation rounds $\delta$ are related as follows:
\begin{align}
\varepsilon\leq\frac{\sum_{i}|c_i|}{\epsilon_0 N_s} \delta\end{align}
where $\epsilon_0 > 0$ is a lower bound on the gradient norm, 
i.e. $\|g\|_2 \geq \epsilon_0$.
\label{lem:rel_error}
\end{lemma}
\begin{proof}
When we perform one execution of the PQC, the state collapses to an eigenstate of the measured observable. Therefore, the finite-shot estimator of the cost function, 
obtained by distributing $N_s$ shots across the $N_o$ Pauli terms, can be expressed as:
\[
f =  \frac{N_o}{N_s}\sum_{i=1}^{N_{\text{o}}}c_i \sum_{k=1}^{N_s/N_o}
 \lambda^{i}_{k}
\]

Let's consider a single eigenvalue $\lambda^i_k$ of the local Pauli $P_i$ measured at the $k$-th round. Due to attacks, in the worst-case scenario, we instead measure a corrupted eigenvalue equal to the most distant w.r.t. the correct $\lambda^i_k$.  One can bound the difference between the correct and the worst-case corrupted eigenvalues (respectively $\lambda^i_k$ and $\hat{\lambda}^i_k$) as $\lambda^i_k-\hat{\lambda}^i_k \leq 2$, since Pauli operators have eigenvalues $+1$ or $-1$.

Due to the blindness property of our delegated protocol, we can assume that on average, the $\delta'$ corrupted executions are uniformly distributed among the $N_o$ Pauli observables, yielding an expected $\delta'/N_o$ corrupted measurements per local expectation value.
Therefore, the absolute difference between the ideal function $f$ and its corrupted version can be bounded as follows:
\begin{align*}
    |f-\hat{f}| \leq
\frac{N_o}{N_s}\sum_{j=1}^{N_{\text{o}}}|c_j\sum_{k=1}^{N_s/N_o}
 (\lambda^{j}_{k}-\hat{\lambda}^j_k)|
 \leq\frac{2\delta'}{N_s}\sum_{j=1}^{N_{\text{o}}}|c_j|
\end{align*}

For each component of the gradient vector obtained with the parameter shift rule we get:
\begin{align*}
    |g_i-\hat g_i|&=\frac{1}{2}|(f_{i+}-\hat{f}_{i+} ) - (f_{i-}-\hat{f}_{i-})| \\ &\leq\frac{(\delta_{i+}+\delta_{i-})}{N_s}\sum_{j=1}^{N_o}|c_j|.
\end{align*}

The total number of corrupted computation rounds per step, distributed across the function evaluations required for all the components of the gradient vector, is expressed as $\delta = \sum_{i=1}^{N_P} \delta_i, \quad \text{where }\delta_i=\delta_{i+}+\delta_{i-}$. 

Given the relative error $\varepsilon = \frac{\|g-\hat{g}\|_2}{\|g\|_2}$, 
the numerator can be bounded using the component-wise absolute 
error bound:
\begin{align*}
    |g_i - \hat{g}_i| \leq \frac{\sum_j|c_j|}{N_s}\delta_i
\end{align*}
\begin{align*}
\|g-\hat{g}\|_2 
= \sqrt{\sum_{i=1}^{N_P}|g_i - \hat{g}_i|^2}
\leq \frac{\sum_j|c_j|}{N_s}\sqrt{\sum_{i=1}^{N_P}\delta_i^2}
\leq \frac{\sum_j|c_j|}{N_s}\delta
\end{align*}
where in the last step we used 
$\sqrt{\sum_i \delta_i^2} \leq \sum_i \delta_i$ 
when all $\delta_i \geq 0$.
Assuming $\|g\|_2 \geq \epsilon_0 > 0$, we obtain:
\begin{align}
\varepsilon \leq \frac{\sum_j|c_j|}{\epsilon_0 N_s}\delta
\end{align}
\end{proof}

Given a desired target relative error $\varepsilon_{th}$, Lemma \ref{lem:rel_error} provides a bound on the tolerated corrupted computation rounds $\delta_{max}$ affecting one Step:

\begin{align}
    \delta \leq \delta_{max} = \frac{\varepsilon_{th}\epsilon_0N_s}{\sum_{j}|c_j|}= \frac{\varepsilon_{th}\epsilon_0d}{2N_P\sum_{j}|c_j|}=wd
\label{eq:delta_max}
\end{align}
where $w= \frac{\varepsilon_{th}\epsilon_0}{2N_P\sum_{i}|c_i|} $.\\

Now that we have a bound on $\delta$, we can proceed with showing verifiability of a MB-DVQA step.

\begin{theorem}
Let $g$ denote the Parameter Shift Rule approximation of $\nabla f$, and assume that $\|g\|_2 \geq \epsilon_0 > 0$.
  An MB-DVQA step following Protocol \ref{prot:VDgrad} computes  $g$ with relative error up to $\varepsilon_{th}$ and probability of failure decreasing exponentially with the number of rounds $n$.
\label{th:ver_step}
\end{theorem}

\begin{proof}
Given a total of $m$ attacked rounds, let $X_d\approx\text{Binomial}(m,\frac{d}{n})$ be the random variable describing the probability of computation rounds being attacked. 
Also, let $X_t \approx \text{Binomial}\left(\left(\frac{m}{n}-\epsilon_1\right)t,\frac{1}{c}\right)$ be the random variable related to the number of attacks detected by the traps in the trapification scheme (with some $\epsilon_1>0$).
Therefore, the probability of failure will be the probability of having more than $\delta_{max}$ corrupted computation rounds and no detected attacks by any trap. This probability can be bounded as:

\begin{align*}
& \text{P(fail)}\leq  \max_{m \in [0, n]} \operatorname{Pr}\left[X_d \geq \delta_{max} \land X_t < 1 \right] \\
&\leq \max \left(\max_{m < nw}\operatorname{Pr}\left[X_d \geq \delta_{max}\right], \max_{m \geq nw }\operatorname{Pr}\left[X_t < 1\right]\right)
\end{align*}

When $m\geq nw$,
considering that computation and test rounds are chosen at random, the number of affected computation rounds $m\frac{d}{n}$ is higher than $\delta_{max}$, leading to $\text{Pr}[X_d > \delta_{max}]=1$. Therefore, we focus on the test rounds term to upper bound the probability of failure.Instead, when $m<nw$, we upper bound the term related to the computation rounds.

\textbf{A).} 
Let's define $X' \approx \text{Binomial}(m,\frac{t}{n})$  to be the probability distribution describing the attacked test rounds. By the law of total probability, 
the probability to have no detected attacks in the test rounds is:

\begin{align*}
\notag \text{Pr}(X_t<1) &= \text{Pr}\left[X_t = 0 | X'\leq k
 \right]\text{Pr}\left[X'\leq k
 \right] \\
&+ \text{Pr}\left[X_t = 0 | X'> k
\right] \text{Pr}\left[X'> k
\right] \\
&\leq \text{Pr}\left[X'\leq k
 \right]+\text{Pr}\left[X_t = 0 | X'> k
 \right]
\end{align*}
where $k= \left(\frac{m}{n}-\epsilon_1\right)t$.

Using Hoeffding's inequality and the fact that $ (1-1/c)^k \leq e^{-k/c}$, we can rewrite the above probability (with $\tau=t/n$):
\begin{align*}
    \text{Pr}(X_t<1)&\leq
    \text{exp}\left(\frac{-2\epsilon_1^2\tau^2n^2}{m}\right)+
    \text{exp}\left(- \left(\frac{m}{n} -\epsilon_1\right)\frac{\tau n}{c}\right) 
\end{align*}

Since $w \leq 1$ 
($m$ always s.t. $m\leq n$), 
we have $n\leq n^2/m \leq n/w$, 
and therefore $\exp\!\left(-2\epsilon_1^2\tau^2\frac{n^2}{m}\right) 
\leq \exp\!\left(-2\epsilon_1^2\tau^2 n\right)$, yielding:
\begin{align*}
\text{Pr}(X_t<1)\leq \text{exp}\left(-2\epsilon_1^2\tau^2 n\right)+ \text{exp}\left(-\frac{\tau n}{c}\left(w -\epsilon_1\right)\right) 
\end{align*}

Both terms, with $w >\epsilon_1$, decay exponentially with the number of copies $n$.

\textbf{B).} When $m<nw$, using Hoeffding's inequality:
\begin{align*}
    \text{Pr}[X_d > \delta_{max}]
    \leq  \text{exp}\left(-\frac{2d^2}{m}\left(\frac{m}{n}-w\right)^2\right)
\end{align*}
Therefore we have a bound on the probability of failure decaying exponentially w.r.t. the number of copies $n$ and the number of test rounds $t$, since $d$ is fixed.

\end{proof}

So far, the verifiability of one step of the MB-DVQA has been shown. However, it might not be clear how to choose the threshold on the tolerated gradient relative error $\varepsilon_{th}$. In the next section, we will show how to appropriately set this bound to ensure both convergence of the optimization and verification of the full MB-DVQA.

\subsection{Verifiable MB-DVQA with Gradient Descent}
\label{sec:varif_vqa}

In this section we explore what happens when we build a verifiable MB-DVQA protocol (Protocol \ref{prot:VDvqa}) using Protocol \ref{prot:VDgrad}. Given limited resources, namely a total number of available rounds, the MB-DVQA protocol runs Protocol \ref{prot:VDgrad} for each step of the optimization. In case a step aborts, the protocol rejects the last gradient computation and re-performs the same step; otherwise it accepts the outcome (with a relative error at most $\varepsilon_{th}$) and updates the variational parameters. 

In the scenario where a delegated VQA is executed with verifiability provided only at the step level, there is no guarantee regarding the final outcome of the variational optimization, namely the final value of $f$. Although highly corrupted steps can be rerun, once the MB-DVQA is terminated, we cannot clearly state if the outcome matches the optimum of the cost function $f^*$. Therefore, such a protocol would have no practical means of knowing whether an execution is successful. To address this gap, we introduce protocol \ref{prot:VDvqa} that achieves verifiability by means of geometric convergence of the optimization procedure, as shown in Theorem \ref{th:ver_vqa}.
This is done by properly tuning $\varepsilon_{th}$ depending on the problem, a task made possible by the result achieved in the previous sections with Theorem \ref{th:ver_step} and Lemma \ref{lem:rel_error}.
Given convergence guarantees, once the protocol terminates, the Client either obtains the desired outcome $f^*$ or receives an abort response, achieving verifiability.

\begin{theorem}
    Given a $\mu$-strongly convex cost function $f$ with $L$-Lipschitz continuous gradient and a global optimum $f^*$, let $g$ denote the Parameter Shift Rule approximation of $\nabla f$ and assume that $\|g\|_2 \geq \epsilon_0 > 0$.
    Protocol \ref{prot:VDvqa} ensures verifiability, provided that the learning rate $\alpha$ and error bound on the gradient $\varepsilon_{th}$ satisfy the following conditions:

\begin{enumerate}
    \item $\alpha L < \frac{2(1-\varepsilon_{th})^2 - \varepsilon_{th}^2}{(1-\varepsilon_{th})(1+\varepsilon_{th})^2},$
    \item $\mu\alpha\left[2(1-\varepsilon_{th})-\alpha L(1+\varepsilon_{th})^2-\frac{\varepsilon_{th}^2}{1-\varepsilon_{th}}\right]< 1 .$
\end{enumerate}
\label{eq:th2_cond}
\label{th:ver_vqa}
\end{theorem}

\begin{algorithm}
 \caption{\raggedright Verifiable MB-DVQA with gradient descent}
 \begin{algorithmic}
    \STATE \textbf{Inputs}: Initial angles $\theta^1$, maximum number of steps $N_{\text{iter}}$, a $c$-coloring, 
           $N_s$, $\varepsilon_{th}$  and learning rate $\alpha$.
    \STATE \textbf{Protocol:}
    \begin{enumerate}
        \item For $k \in [1,...,N_{\text{iter}}]$, run Protocol \ref{prot:VDgrad} to compute $\hat{g}(\theta^{k})$:
        \begin{itemize}
            \item If it accepts: $\theta^{k+1}=\theta^{k}-\alpha \hat{g}(\theta^{k})$.
            \item If it aborts: $\theta^{k+1}=\theta^{k}$ (i.e. re-run the same step).
        \end{itemize}  
        \item Set the final angles: $\theta^* \leftarrow \theta^{N_{\text{iter}}}$.
        \item If the optimization hasn't converged, \textbf{aborts}; otherwise \textbf{accepts} returning $\hat{f}^*$.
    \end{enumerate}  
 \end{algorithmic}
  \label{prot:VDvqa}
\end{algorithm}

\begin{proof}[Sketch of Proof]
   As explained in detail in Appendix \ref{sec:grad_conv_2}, a VQA using Stochastic Gradient Descent to optimize a function $f$ can converge in the presence of adversarial attacks provided that the gradient relative error remains sufficiently bounded. If no such bound exists, multiple steps with large gradient errors may occur, potentially jeopardizing the optimization and leading to divergence. To prevent this and ensure convergence, Protocol \ref{prot:VDvqa} integrates Protocol \ref{prot:VDgrad}, which detects steps where the relative error exceeds the threshold $\varepsilon_{th}$. Such steps are then re-executed until the conditions of the Theorem hold. 
   These inequalities, derived in the full proof (Appendix \ref{sec:grad_conv_2}), are related to standard conditions required for gradient descent optimization, namely having $0<\gamma<1$ such that for each optimization step $k$:
   \[
   \mathbb{E}[f^{k+1}]-f^* \leq \gamma (\mathbb{E}[f^{k}]-f^*)
   \]
   Notably, when $\varepsilon_{th}\rightarrow0$, the inequalities above reduce to the standard gradient descent conditions for convergence, having $\gamma=1-\mu\alpha(2-\alpha L)$.
   This guarantees the average decrease of the function required for convergence and also bounds the final error neighborhood in the presence of shot-noise, albeit at the cost of an overhead of steps due to repetitions.
   Finally, we associate convergence guarantees with verifiability: since the Client has access to classical gradient information at each step $k$, the protocol either outputs the correct result, namely the optimal value of the cost function achieved due to convergence, or aborts.
\end{proof}

While VQA landscapes are generally non-convex, convexity-based analyses—common also in classical machine learning—can still provide useful insight, for example when initialization or ansatz design places the optimizer in more regular regions of the landscape. In the above theorem, we adopt the strong convexity assumption as a simplified analytical setting to enable convergence analysis \cite{Sweke2020stochasticgradient}. This assumption also guarantees that the lower bound on the gradient norm $\epsilon_0$ in Lemma~\ref{lem:rel_error} is sufficiently large and therefore provides a meaningful bound for the gradient relative error $\varepsilon$. 
Our assumptions do not cover scenarios where the cost landscape is affected by barren plateaus or local minima; addressing fully realistic, problem-specific landscapes and alternative optimization methods is left for future work.

\section{Simulations}
We start by benchmarking the performance of Protocol \ref{prot:VDvqa}, to demonstrate its effectiveness. We do this
by trying to find the ground state of the Transverse Field Ising Model (TFIM), a fundamental quantum spin model describing a system of interacting spins (typically on a lattice):
\begin{align*}
H=h\sum_iX_i + \sum_{<i,j>}Z_iZ_j
\end{align*}
where $<i,j>$ determines the edges of the lattice and $h$ is the strength of the external transverse field.
We express attacks as perturbations on the gradient vector for each step $k$ (in the context of Protocol \ref{prot:VDvqa}):
\begin{align}
\hat{g}_i(\theta^k) = g_i(\theta^k)+ \nu_i^k \end{align}
Each $\nu_i^k$ is drawn independently from a uniform distribution $
 \mathcal{U}(-\Delta, \Delta)
$, where $\Delta$ is a parameter tuning the impact of the perturbation. The perturbation vector is defined as:
\[\nu^k =
\begin{cases}
(\nu_1^k, \nu_2^k, \ldots, \nu_{N_P}^k), & \text{with probability } p \\
(0, 0, \ldots, 0), & \text{with probability } 1 - p
\end{cases}
\]

The simulation of the computation rounds in the absence of quantum hardware noise is carried out using the Qrisp programming framework \cite{seidel2024qrispframeworkcompilablehighlevel}.
For now we assume the existence of traps allowing for a negligible probability of failure and leave the simulation of the test rounds for later. In this way, we can simulate the trap detection by evaluating the relative error. Protocol \ref{prot:VDvqa} re-performs heavily attacked steps without aborting the full VQE and accepts steps that fall below the error threshold.

We show examples for two-dimensional TFIMs using $p=0.7$ and an ansatz with $N_L$-layers and $n$ qubits, each layer $U_l$ is given by: 
\[
U_l(\theta_l) = \left(\prod_{i=1}^{n-1} \text{CNOT}_{i, i+1}\right) \left(\bigotimes_{i=1}^{n} R_y(\theta_l^i)\right)
\]

\begin{figure}[h!]
\includegraphics[width=0.45\textwidth]{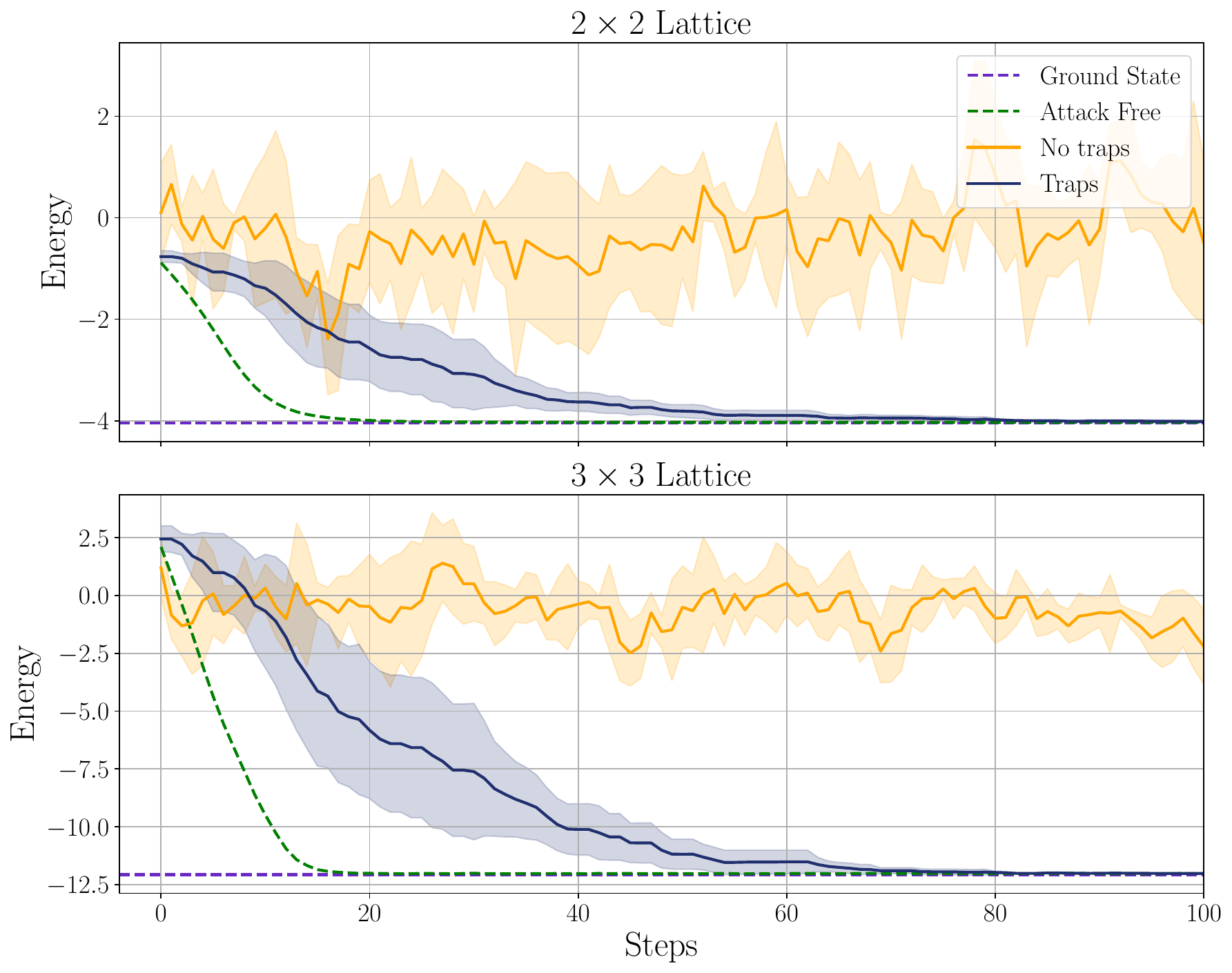}
    \caption{Average performance over 5 runs comparing the proposed protocol and a standard VQE in presence of the described attacks with $p=0.7$. We tackle the following two-dimensional TFIM lattices: $2 \times 2$ and $3 \times 3$ with $h=0.2$. The dashed green line (Attack Free) refers to an ideal attack-free scenario, the yellow line (No traps) is related to a non-verifiable delegated VQE in presence of attacks  and finally, the proposed protocol is shown in blue (Traps). We use $N_s=1000$, $N_L=2$ layers and a learning rate $\alpha=0.2$ for the optimization.}
    \label{fig:TFIM2D}
\end{figure}

\begin{figure}[h]
\includegraphics[width=0.48\textwidth]{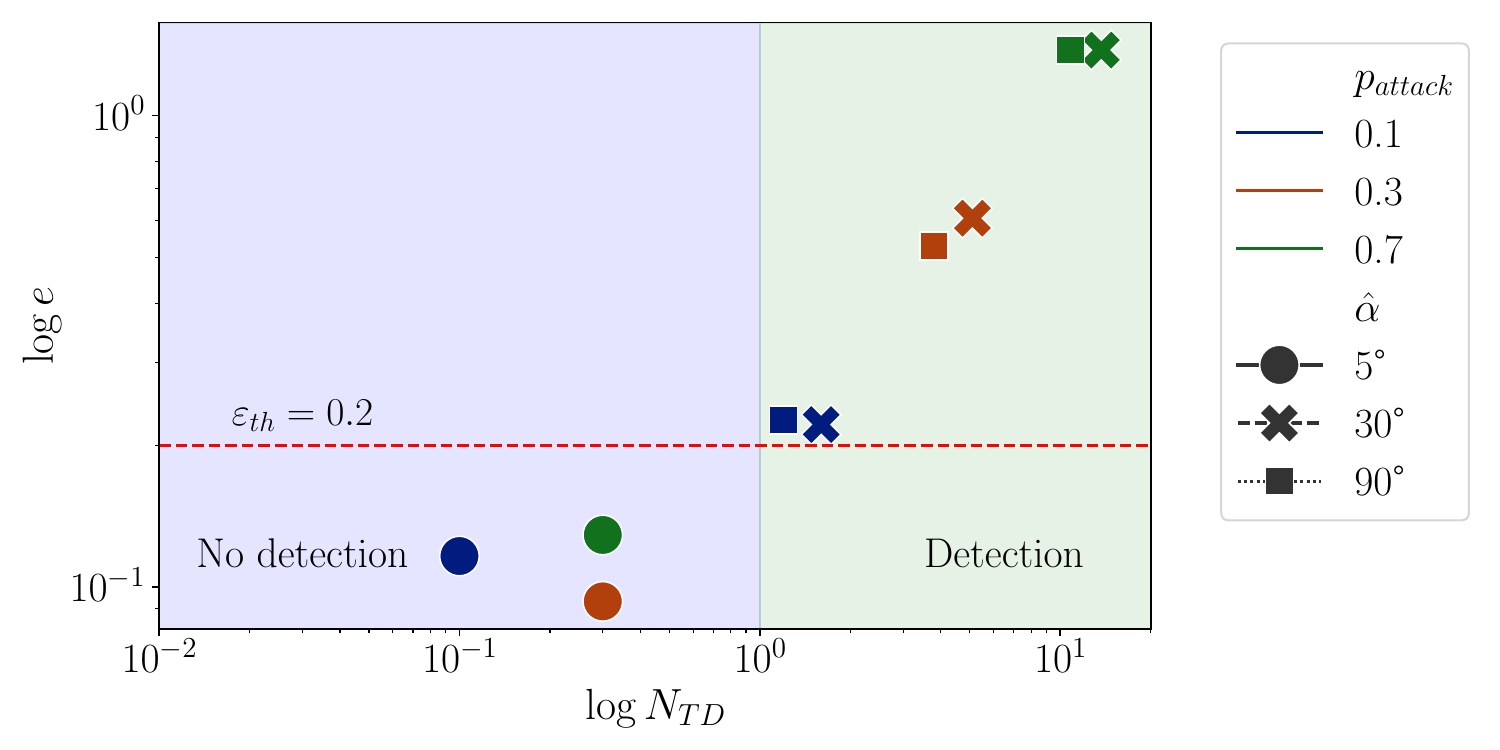}
    \caption{ Plot of relative error  $\varepsilon$ versus number of trap detections $N_{TD}$ in one execution of Protocol \ref{prot:VDgrad} (in logarithmic scale). Here, we choose $d=4000$ and $t=50$. $p_{attack}$ is the probability that a round is attacked and $\hat{\alpha}$ shows examples of a degree shift applied by an attacker on the measurement angles of the output qubits. The points represent average scores computed over 10 runs for each setting. As expected, when the relative error exceeds the threshold $\varepsilon_{th}$ (shown with the red dashed line), the attack is detected.}
    \label{fig:step_figure}
\end{figure}

From Fig. \ref{fig:TFIM2D} we can see how a standard VQE using gradient descent diverges due to attacks, while the proposed approach, as expected, turns out to be robust and converges to the optimal solution with an overhead of steps with respect to the ideal attack-free case.
Using shot-level verification would significantly increase the number of aborted shots over the full VQE run, extending the runtime. In contrast, our approach tolerates a reasonable fraction of corrupted shots, while enabling convergence within a reasonable timeframe.

It is important to emphasize that our protocol does not depend on the aforementioned attack model. The latter has been designed to illustrate the performance of the proposed work.
As reported in the previous section, Protocol \ref{prot:VDvqa} is general and robust, providing tools to adjust the number of traps to ensure convergence, independent of the attack model.

To complement the numerical analysis, we provide simulations of Protocol \ref{prot:VDgrad} tackling a 2-qubit TFIM. In particular, we make use of Veriphix \cite{Abdul_Sater_Veriphix}, a library for verifiable blind quantum computing based on the MBQC programming framework Graphix \cite{sunami_graphix_2022}.
In Fig. \ref{fig:step_figure} we show how the number of test rounds $t$ can be tuned to detect highly corrupted steps depending on $\varepsilon_{th}$, the threshold value required by Protocol \ref{prot:VDvqa} to ensure convergence.
The Python code used to reproduce the simulations is openly available at \cite{githubcode}.

\section{Conclusion}
This work offers valuable insights into the future of delegated quantum computation in the NISQ era and beyond, by proposing an end-to-end verifiable delegated variational quantum protocol. The theoretical results are further validated with simulations that tackle the Transverse Field Ising Model. Although our work is currently based on gradient descent optimization, it is readily adaptable to alternative optimization techniques. Specifically, Protocol \ref{prot:VDgrad} can be modified to accommodate different optimizers, leading to appropriately revised versions of Lemma \ref{lem:rel_error} and Theorem \ref{th:ver_step}. 
On the other hand, with regard to the verifiability of the overall VQA, an optimizer-specific Theorem should be provided to guarantee convergence and consequentially, verifiability.
Furthermore, in our protocol, we assume a fixed relative error threshold for all the steps. In future work, the threshold across the steps could be varied depending on the current gradient value and accordingly the number of traps required for each step could be adjusted, thus reducing the total number of delegated computations. Another extension of this work could involve an improvement of the trapification procedure, for instance exploiting the scheme based on stabiliser testing proposed in \cite{Kapourniotis_2024}. 
Finally, although circuit-based DVQA protocols have already been proposed (such as the one by Li et al. \cite{li2023circuitDVQAQHE} based on quantum homomorphic encryption) they currently only provide blindness. It remains an open and interesting question how verification could be incorporated in such schemes, either through adaptations of our approach or via dedicated techniques like accreditation \cite{Ferracin_2021}.

\begin{acknowledgments}
M.I. and A.P. acknowledge support from the Federal Ministry for Economic Affairs and Climate Action (project Qompiler, grant No: 01MQ22005A), the DFG (Emmy Noether programme, grant No. 418294583 and SPP2514), the BMBF (project tubLANQ.0, grant No. 16KISQ087K) and the Berlin Quantum Alliance. P.W. acknowledges support by EPSRC grant
EP/T001062/1.
\end{acknowledgments}

\bibliographystyle{apalike}
\bibliography{refs}

\begin{widetext}
\newpage
\section{APPENDIX}
\subsection{Useful Lemmas}

\begin{lemma}
(Hoeffding's inequality)
Given a binomial distribution $X \approx \text{Binomial}(n,p)$ an upper bound to the cumulative distribution function can be obtained for $k\leq np$:

\begin{align*}
    \Pr(X \le k) \leq \exp\left(-2 n\left(p-\frac{k}{n}\right)^2\right)
\end{align*}
\label{lem:hoeffding}

\end{lemma}

\begin{lemma}
    (Lemma, page 870 in \cite{Polyak1963GradientMF}) Given $f:\mathbb{R}^d\rightarrow\mathbb{R}$ with optimum $f^*$ and $\mu$-strongly convex, it holds that:
\[
    2\mu (f(x)-f^*)\leq \|\nabla f(x)\|^2, \forall x \in \mathbb{R}^d
\]
\label{lem:polyak}
\end{lemma}

\begin{lemma}
    (Lemma 4.2 in \cite{bottou2018optimizationmethodslargescalemachine})
    The iterations of the Stochastic Gradient Descent (with learning rate $\alpha$) for a $L$-Lipschitz continuous gradient function $f$ satisfy:
    \[
    \mathbb{E}[f(\theta^{k+1})]-f(\theta^k) \leq -\alpha \nabla f(\theta^k)^T \mathbb{E}[g(\theta^k)] +\frac{\alpha^2L}{2}   \mathbb{E}[\|g(\theta^k)\|_{2}^2]
    \]
    where $g$ is the estimation of $\nabla f$.
\label{lem:bottou}
\end{lemma}

\subsection{Stochastic Gradient Descent Convergence with shot-noise and attacks}
\label{sec:grad_conv_2}
Here we provide convergence guarantees for a VQA using gradient descent in presence of adversarial attacks. 
In an ideal gradient descent optimization, one has access to the real gradient value $\nabla f$, therefore the update rule for each optimization step $k$ is:
\[
\theta^{k+1}=\theta^{k}-\alpha\nabla f(\theta^{k})
\]
Nevertheless, when running gradient descent as optimizer in a practical setting the stochasticity of the estimation of  $\nabla f$ needs to be taken into account. In particular, by executing a VQA with a NISQ device the stochasticity source is the so-called shot-noise, namely the estimation error due to limited amount of measurements available for the computation of the expectation values. In this scenario, we rewrite the update rule as follows:
\[
\theta^{k+1} = \theta^{k} - \alpha\, g(\theta^{k}),
\]
where \( g(\theta^{k}) \) is a stochastic (shot-noise) estimate of the true gradient 
\( \nabla f(\theta^{k}) \).  Here we assume that $\mu_g=\mathbb{E}[g]= \nabla f$ and that the stochastic gradient estimator \( g(\theta^{k}) \) has covariance $\mathrm{Cov}[g(\theta^{k})] = \frac{\Sigma}{N_s}$; hence the total variance of the gradient estimate is
\[
\sigma_g^2 = \frac{\mathrm{Tr}[\Sigma]}{N_s} 
= \frac{1}{N_s} \sum_{i=1}^{N_P} \Sigma_{ii}
\]
where \( \Sigma_{ii} = \mathrm{Var}[g_i] \) denotes the variance of the $i$-th component.
By introducing in our analysis adversarial attacks we introduce the 
`corrupted' gradient $\hat{g}$, having:
\[ \theta^{k+1}=\theta^{k}-\alpha \hat{g}(\theta^{k})\]\\

Thanks to our Protocol and Lemma \ref{lem:rel_error} we have:
\begin{equation*}
\|\hat{g}(\theta^k) - g(\theta^k)\|_2 \le  \varepsilon, \|g(\theta^k)\|_2,
\end{equation*}
Then we define $\nu^k := \hat{g}(\theta^k) - g(\theta^k)$, so that $\|\nu^k\|_2 \le \varepsilon \|g(\theta^k)\|_2$.

In order to investigate the convergence and apply Lemma \ref{lem:bottou} we need to compute $ \nabla f^T \mathbb{E}[\hat{g}]$ and $\mathbb{E}[\|\hat{g}\|_{2}^2]$.

By linearity of expectation and Cauchy--Schwarz:
\begin{align}
\mathbb{E}[\nabla f^\top \hat{g}] 
&= \nabla f^\top \mathbb{E}[g] + \nabla f^\top \mathbb{E}[\nu] \notag\\
&= \|\nabla f\|_2^2 + \nabla f^\top \mathbb{E}[\nu] \notag\\
&\ge \|\nabla f\|_2^2 
- \|\nabla f\|_2 \|\mathbb{E}[\nu]\|_2 \notag\\
&\ge \|\nabla f\|_2^2 
- \|\nabla f\|_2 \, \mathbb{E}\|\nu\|_2 \notag\\
&\ge \|\nabla f\|_2^2 
- \varepsilon \|\nabla f\|_2 \, \mathbb{E}\|g\|_2 \notag\\
&\ge \|\nabla f\|_2^2 
- \varepsilon \|\nabla f\|_2 \sqrt{\mathbb{E}\|g\|_2^2} \notag\\
&= \|\nabla f\|_2^2 
- \varepsilon \|\nabla f\|_2 
\sqrt{\|\nabla f\|_2^2 + \sigma_g^2} \notag\\
&\ge (1-\varepsilon)\|\nabla f\|_2^2 
- \varepsilon \|\nabla f\|_2 \sigma_g
\end{align}
where we used:
(i) Jensen inequality $\|\mathbb{E}[\nu]\| \le \mathbb{E}\|\nu\|$,
(ii) the error bound $\|\nu\|_2 \le \varepsilon\|g\|_2$,
(iii) Cauchy--Schwarz $\mathbb{E}\|g\| \le \sqrt{\mathbb{E}\|g\|^2}$, (iv) $\mathbb{E}[\|g\|_2^2] = \|\nabla f\|_2^2 + \sigma_g^2$,
and (v) $\sqrt{a^2+b^2} \le a + b$ for $a,b \ge 0$.

\begin{align}
\mathbb{E}[\|\hat{g}\|_2^2] 
&= \mathbb{E}[\|g+\nu\|_2^2] \notag\\
&= \mathbb{E}[\|g\|_2^2] + 2\mathbb{E}[g^\top\nu] + \mathbb{E}[\|\nu\|_2^2] \notag\\
&\leq \mathbb{E}[\|g\|_2^2] + 2\varepsilon\mathbb{E}[\|g\|_2^2] + \varepsilon^2\mathbb{E}[\|g\|_2^2] \notag\\
&= (1+\varepsilon)^2\mathbb{E}[\|g\|_2^2] \notag\\
&= (1+\varepsilon)^2(\|\nabla f\|_2^2 + \sigma_g^2)
\end{align}
Now, starting from Lemma \ref{lem:bottou}:

\begin{align}
\notag
\mathbb{E}[f(\theta^{k+1})]-f(\theta^k)
&\leq -\alpha\mathbb{E}[\nabla f^\top \hat{g}] 
+ \frac{\alpha^2 L}{2}\mathbb{E}[\|\hat{g}\|_2^2] \\
&\leq -\alpha(1-\varepsilon)\|\nabla f\|_2^2 
+ \alpha \varepsilon\|\nabla f\|_2\sigma_g
+ \frac{\alpha^2 L}{2}(1+\varepsilon)^2(\|\nabla f\|_2^2+\sigma_g^2) \notag\\
&\leq -\frac{1}{2}\alpha\left[2(1-\varepsilon) - 
\alpha L(1+\varepsilon)^2\right]\|\nabla f\|_2^2
+ \alpha \varepsilon\|\nabla f\|_2\sigma_g
+ \frac{\alpha^2 L}{2}(1+\varepsilon)^2\sigma_g^2
\end{align}

To handle the cross-term $\alpha \varepsilon \|\nabla f\|_2 \sigma_g$ 
and avoid dependence on the iterates, 
we apply the following inequality 
$ab \le \frac{a^2}{2\lambda} + \frac{\lambda b^2}{2}$ \footnote{From $0 \leq (a-b)^2$ it follows that $ab \leq \frac{a^2+b^2}{2}$.
For any $\lambda>0$, replacing with $a'=a/\sqrt{\lambda},b'=b\sqrt{\lambda} $ yields: $ab \leq \frac{a^2}{2\lambda}+\frac{\lambda b^2}{2}$}  
(valid for any $\lambda > 0$ and $a,b\geq0$) obtaining
\begin{align}
\alpha \varepsilon \|\nabla f\|_2 \sigma_g
\le
\frac{\alpha \varepsilon}{2\lambda}\|\nabla f\|_2^2
+
\frac{\alpha \varepsilon \lambda}{2}\sigma_g^2 .
\end{align}
Choosing $\lambda = \frac{1-\varepsilon}{\varepsilon}$ and
substituting into the previous inequality gives
\begin{align}
\notag
\mathbb{E}[f(\theta^{k+1})]-f(\theta^k)
&\le -\frac{\alpha}{2} \left[2(1-\varepsilon)-
\alpha L(1+\varepsilon)^2-\frac{\varepsilon^2}{1-\varepsilon}\right]\|\nabla f\|_2^2\\
&\quad+\left(
\frac{\alpha(1-\varepsilon)}{2}
+
\frac{\alpha^2 L}{2}(1+\varepsilon)^2
\right)\sigma_g^2 .
\end{align}
The descent coefficient is positive whenever
$
\alpha L < \frac{2(1-\varepsilon)^2 - \varepsilon^2}{(1-\varepsilon)(1+\varepsilon)^2},
$ which is the condition stated in the main text.
Applying Lemma~\ref{lem:polyak}, adding $f(\theta^k)-f^*$ to both terms and applying the total expectation, yields:
\begin{align}
\notag \mathbb{E}[f(\theta^{k+1})]-f^*
&\leq \underbrace{\left(1-
\mu\alpha\left[2(1-\varepsilon)-\alpha L(1+\varepsilon)^2-\frac{\varepsilon^2}{1-\varepsilon}
\right]
\right)
}_{\gamma}
\left(\mathbb{E}[f(\theta^k)]-f^*\right)
\\
&\quad
+
\underbrace{
\left(
\frac{\alpha(1-\varepsilon)}{2}
+
\frac{\alpha^2 L}{2}(1+\varepsilon)^2
\right)
\sigma_g^2
}_{B}
\label{eq:conv_noise}
\end{align}
where $B$ is a constant independent of $k$.

Guarantees for convergence are met having $0<\gamma<1$, namely when
\begin{align}
\alpha L < \frac{2(1-\varepsilon)^2 - \varepsilon^2}{(1-\varepsilon)(1+\varepsilon)^2},
\qquad
\mu\alpha
\left[
2(1-\varepsilon)
-
\alpha L(1+\varepsilon)^2
-
\frac{\varepsilon^2}{1-\varepsilon}
\right]
< 1 .
\label{eq:conv_cond}
\end{align}
When $\varepsilon\rightarrow 0$ these reduce to the standard gradient descent 
conditions $\alpha L < 2$, recovering the original result.
Protocol~\ref{prot:VDvqa} allows to keep $\varepsilon$ bounded and tunable 
to satisfy such requirements.

Without shot-noise, namely when $\sigma_g^2\rightarrow 0$, strict convergence is achieved.
Nevertheless,  we cannot have a strict convergence due to the shot-noise variance, but rather a convergence to a neighborhood of $f^*$.

To define such a neighborhood let's rearrange the inequality 
\eqref{eq:conv_noise} as

\begin{align*}
z_{k+1}\leq \gamma z_{k} + B    
\end{align*}

where $z_{k}= \mathbb{E}[f(\theta^{k})]-f^*$. When $k\rightarrow \infty$, then $z_{\infty}(1-\gamma)\leq B$, 
so the final error is bounded as

\begin{align*}
z_\infty \leq \frac{B}{1-\gamma} =
\frac{
\left(
\frac{\alpha(1-\varepsilon)}{2}
+
\frac{\alpha^2 L}{2}(1+\varepsilon)^2
\right)\sigma_g^2
}
{
\mu\alpha
\left[
2(1-\varepsilon)
-
\frac{\varepsilon^2}{1-\varepsilon}
-
\alpha L(1+\varepsilon)^2
\right]
}
\end{align*}

The condition achieved for $\gamma<1$ allows to have a non-divergent error $z_{\infty}$.
We can also derive how many iterations we need to achieve a final error $\epsilon$.
Assuming negligible $\sigma_g^2$ and having $\gamma=(1-\rho)$ we can write:
\begin{align*}
\mathbb{E}[f(\theta^{k+1})]-f^*&\leq (1-\rho)^k  [ \mathbb{E}[f(\theta^{1})]-f^*]
\leq \exp(-k\rho)[\mathbb{E}[f(\theta^{1})]-f^*]
\end{align*}
since  $(1-\rho)^k\leq \exp(-k\rho)$.
To have a final error $\mathbb{E}[f(\theta^{k+1})]-f^*\leq \epsilon$, only $O(\log(1/\epsilon))$ iterations are required:
\begin{align*}
   k\geq \frac{1}{\rho} \log\left(\frac{\mathbb{E}[f(\theta^{1})]-f^*}{\epsilon}\right)
\end{align*}
It's important to notice that the iterations' scaling doesn't take into account the re-performed steps, considering instead only the iterations having a proper bound on the error variable.

Summarizing, to achieve convergence, having $L$ and $\mu$ fixed by the problem, is it enough to set properly the learning rate $\alpha$ and the term $\varepsilon$. The latter can be modified by tuning the number of traps of our step level protocol and Lemma \ref{lem:rel_error}.

\end{widetext}
\end{document}